\documentclass[journal,twoside,web]{ieeecolor}
\usepackage{geometry}
\usepackage{generic}
\usepackage{cite}
\usepackage{amsmath,amssymb,amsfonts}
\usepackage{algorithmic}
\usepackage{graphicx}
\usepackage{textcomp}

\usepackage{graphicx}
\usepackage{subfigure}
\usepackage{epsfig} 
\usepackage{cite}

\usepackage{lcsys}

\newtheorem{theorem}{Theorem}[section]

\newtheorem{proposition}[theorem]{Proposition}

\newtheorem{definition}[theorem]{Definition}

\newtheorem{assumption}[theorem]{Assumption}
\newtheorem{example}[theorem]{Example}
\newenvironment{proofOf}[1]{%
  \par\noindent\textit{Proof of #1.}\quad\ignorespaces
}{%
  \hfill$\blacksquare$\par
}

\allowdisplaybreaks

\begin{document}

\title{Learning Empirical Evidence Equilibria under Weak Environmental Coupling}

\author{\makebox[\textwidth][c]{Aya~Hamed,~Jason~R.~Marden,~Jeff~S.~Shamma}}

\maketitle

\begin{abstract}
Strategic multi-agent systems are fundamentally characterized by decentralization, uncertainty, and ambiguity. Agents operating under limited observations will often need to make decisions based on simplified internal models of the environment, reflecting bounded rationality in both computational capacity and environmental knowledge. The Empirical Evidence Equilibrium (EEE) framework explicitly accounts for these limitations by modeling each agent as forming a potentially misspecified belief derived from signals obtained through partial observations of the environment. The resulting equilibrium concept captures the system's steady state under bounded rationality and decentralization. In this work, we study games in which the environment dynamics are driven jointly by exogenous factors and agents' actions. We analyze agent behavior under Q-value iteration where each agent independently forms a belief model, computes Q-values, and derives a greedy strategy, yet the collective actions of all agents jointly shape the environment each agent faces at the next stage. We prove that despite this decentralization, an EEE emerges from the joint dynamics when the coupling between agents' actions and the environment is sufficiently weak. We further extend this result to softmax policies, establishing a contraction result under a sufficient coupling condition.
\end{abstract}

\IEEEpeerreviewmaketitle

\thispagestyle{empty}
\section{Introduction}
\IEEEPARstart{F}{rom} the complex dynamics of human societies and economic markets to the coordination of autonomous vehicles and robotic swarms, multi-agent systems are inherently large-scale, decentralized, and uncertain. Agents therein are often strategic and self-interested, interacting repeatedly with both the environment and each other, with outcomes shaped by collective behavior and the evolving system state. Stochastic games provide a powerful formalism for such settings, yet agents typically operate under limited observations and uncertain knowledge of both the environment dynamics and others' strategies. This leads to decision-making under partial observability, rendering each agent's problem a Partially Observable Markov Decision Process (POMDP).
\begingroup
\renewcommand{\thefootnote}{}
\footnote{This work is supported by AFOSR grants \#FA9550-25-1-0245 and \#FA9550-21-1-0203, and NASA grant \#103215.}
\footnote{Aya Hamed and Jason R. Marden are with the Department of Electrical and Computer Engineering, University of California, Santa Barbara, CA, ayah@ucsb.edu, jrmarden@ece.ucsb.edu. Jeff S. Shamma is with the Department of Industrial and Enterprise Systems Engineering, the Grainger College of Engineering, University of Illinois at Urbana-Champaign, IL, jshamma@illinois.edu}
\endgroup

Partial observability and stochasticity are inherent characteristics of a wide range of engineering applications. However, POMDPs are well known for their computational intractability. In addition, they rely on the assumption that agents possess a perfect model of both the environment and their opponents, an assumption that rarely holds in practice. Motivated by these modeling limitations, Dudebout and Shamma introduced a new framework, the {Empirical Evidence Equilibrium (EEE)} \cite{r39,r40}, that explicitly accounts for agents’ bounded rationality in terms of the agents' capabilities as well as their knowledge of the environment, leading to a more realistic model of agent behavior and equilibrium concept. 

Several related frameworks relax full-rationality assumptions through simplified modeling, though they differ in their assumptions and applicable settings. Mean Field Games \cite{LL2007, HMC2006} provide a scalable framework for large populations of interacting agents by replacing strategic interactions with a representative mean-field signal. The Berk–Nash Equilibrium \cite{EP2016} allows agents to hold misspecified subjective models, where equilibrium beliefs are the best-fitting distribution within each agent's model class given observed outcomes. Subjective Equilibrium under Beliefs of Exogenous Uncertainty \cite{AY2023} models games where agents incorrectly perceive endogenous environment variables as exogenous.  The Empirical Evidence Equilibrium framework extends this landscape by accommodating agents with varying levels of rationality who receive partial, noisy signals from the environment and form simplified belief models accordingly. Equilibrium is attained when each agent's strategy is optimal with respect to their belief model, and their model is consistent with the empirical distribution of observed signals. We adopt this framework in the present work. 

The setting of interest in this paper is stochastic games in which the environment dynamics are governed by both exogenous factors and agents' actions, with the latter having a bounded effect on the dynamics. Settings of partially controllable environments have been studied in the MDP literature, motivated by a range of real-world applications. The idea of decomposable MDPs was formalized in \cite{NS2019}, where state variables are partitioned into an endogenously controlled component and an exogenously evolving one. This structure is motivated by applications such as capacity portfolio management, where a firm's inventory level is an endogenous state while commodity prices evolve exogenously, and fishery management, where the fish population is jointly shaped by environmental conditions and the harvesting decisions of the manager. The study of MDPs with exogenous factors in the context of Reinforcement Learning (RL) is further motivated in \cite{DTC2018, TD2023}, noting that in settings such as wireless cellular networks, system performance is determined by both endogenous variables (e.g., cell tower parameters) and exogenous ones (e.g., customer demand driven by the number and spatial distribution of users). These works identify the slowing effect of such exogenous variation on RL algorithms and propose methods to mitigate it. Most recently, \cite{MSR2026}, motivated by the dynamics of real-world problems such as trading and reservoir management, similarly explores the performance of RL in MDPs where some state variables are endogenously controlled while the rest evolve exogenously, exploiting this structure to yield improved learning guarantees.

The present work addresses a related setting in stochastic games; rather than partitioning state variables into endogenous and exogenous components, we allow agents' actions to influence the transition dynamics globally, but only as a bounded deviation from a baseline uncoupled kernel, a structure motivated by the same class of real-world applications. Following the framework of bounded rationality and decision-making under local information, we study independent learning dynamics, specifically, Q-value iteration, in which each agent forms an internal model, computes Q-values, and derives an optimal strategy independently, yet the collective actions of all agents jointly shape the environment each agent faces at the next stage. We prove the emergence of an EEE through these dynamics, and establish convergence guarantees when the environment is weakly coupled; that is, when the transition kernel is sufficiently close to the uncoupled baseline.
\section{The Stochastic Game Framework}
We adopt the stochastic games framework of \cite{r39}, defining the following setup. The agents' actions affect the environment dynamics and generate agent-specific signals. Each agent forms a simplified model of these signals and receives a stage payoff depending on their state, action, and observed signal.
\label{subsec:NaturalSetup}

Let $\mathcal{I}$ be a finite set of agents. Each agent $i \in \mathcal{I}$ has a finite state space $\mathcal{X}_i$, a finite action space $\mathcal{A}_i$, and observes signals from a finite set $\mathcal{S}_i$. The environment itself has a finite state space $\mathcal{W}$. The joint state and action profiles of all agents are denoted by ${x} = (x_i)_{i\in \mathcal{I}}$ and ${a} = (a_i)_{i\in \mathcal{I}}$, respectively. 

The environment and agents' state dynamics, denoted by $\mathbf{N}$, evolve as follows:
\[
w^+ \sim \Phi(w,a),
\qquad
s_i \sim M_i(w),
\qquad
x_i^+ \sim \varphi_i(x_i,s_i,a_i),
\]
where $\sim$ denotes drawing from the corresponding distribution. For compactness, we will henceforth write $\Phi_{w,w^+}(a) \triangleq \mathbb{P}(w^+ \mid w, a)$, viewing $\Phi(a)$ as a stochastic matrix indexed by $(w, w^+)$, with $a$ as its argument.
Each agent receives a stage payoff defined by the function 
$g_i : \mathcal{X}_i \times \mathcal{A}_i \times \mathcal{S}_i \to \mathbb{R}$ and has a discount factor of $\delta_i.$

However, each agent only observes the signals $\mathcal{S}_i$ from the environment and is unaware of how they are generated and instead forms a misspecified model $\mu_i,$ described below, upon which it bases its decisions. An agent $i$ maintains a memory state $z_i$ in a finite set $\mathcal{Z}_i$, the structure of which may vary in sophistication. This memory is updated through a local function $l_i$, which takes as input the current memory state and signal.
An agent's presumed dynamics $\mathbf{M}_i$ is described as:
\[
s_i \sim \mu_i(z_i,x_i), \qquad
z_i^+ = l_i(z_i,s_i), \qquad
x_i^+ \sim \varphi_i(x_i,s_i,a_i).
\]
An agent chooses actions via a stationary policy 
$a_i \sim \sigma_i(z_i,x_i).$
The EEE equilibrium notion was defined in \cite{r39} to reflect the bounded rationality of agents and their limited observations of the environment. An equilibrium is achieved when agents act optimally given their formed models and presumed dynamics $\mathbf{M}_i$ and, at the same time, these models are consistent with the long-run observations induced by these optimal strategies and the true environment dynamics $\mathbf{N}$. We recall the 
EEE definition from \cite{r39} as applied to our framework, under the following regularity condition.

\begin{assumption}
     The Markov chain induced by $\mathbf{N}$ and any strategy ${\sigma}$ is ergodic and admits a unique stationary distribution.
    \label{StochasticAssumption}
\end{assumption}
\begin{definition}
Consider a set of agents $\mathcal{I}$, each with strategy $\sigma_i$ and model $\mu_i$. The pair of profiles $({\sigma}, {\mu})$ is an \textbf{EEE} if, for every agent $i \in \mathcal{I}$, the following conditions hold:  

 \noindent\textbf{Optimality:} The strategy $\sigma_i$ is greedy with respect to the agent's Q-value function. That is, for any state $(z_i, x_i)$: \[ \sigma_i(z_i, x_i) \in \arg\max_{a'_i \in \mathcal{A}_i} Q_i(z_i, x_i, a'_i) \] where the Q-value and V-value functions are defined by the Bellman equations under agent $i$'s own dynamical model $\mathbf{M}_i,$ with $\mu_i$ as the signal generation model: \begin{align*} Q_i(z_i, x_i, a_i) &= \mathbb{E}_{\mathbf{M}_i} \left[ g_i(x_i, a_i, s_i) + \delta_i V_i(z_i^+, x_i^+) \right], \\ V_i(z_i, x_i) &= \max_{a'_i \in \mathcal{A}_i} Q_i(z_i, x_i, a'_i). \end{align*}
 
 \noindent\textbf{Consistency:} The agent's model $\mu_i$ is consistent with the \emph{long-run signal frequencies} generated by $\mathbf{N}$ when all agents follow the joint strategy ${\sigma}$. That is:
\[\mu_i(z_i, x_i)[s_i] = \lim_{h \to \infty} \mathbb{P}_{\mathbf{N}, \sigma}(s_i^h=s_i \mid
z_i^h = z_i,\ x_i^h = x_i),\] where the limit is well-defined under Assumption~\ref{StochasticAssumption} and equals the conditional frequency of signal $s_i$ under the stationary distribution of the joint state process $(w, z_i, x_i)$ induced by $\mathbf{N}$ and $\sigma$. In practice, this can be approximated by the empirical signal frequencies observed over a long horizon.
 
 \end{definition}

In this work, we investigate whether boundedly rational agents operating under the EEE framework are able to reach an EEE through independent distributed dynamics. In \cite{thesis}, finite static games were shown to be a special case of the EEE framework, with Nash equilibria (NE) being equivalent to EEE in that setting. Yet even in this simplified case, establishing dynamics that converge to NE in full generality is unresolved. To address this, structure is often imposed on these games to guarantee convergence. This motivates our focus on the study of agents' behavior in structured stochastic games. In particular, we consider \textit{weakly-coupled} stochastic games, a model that, as noted earlier, captures many systems of practical interest that are governed by both external factors and agents' actions, where the influence of agents' actions' is bounded. We formally define a stochastic game with a {weak coupling} as follows:

\begin{definition}
\label{weak}
A stochastic game is \textbf{weakly-coupled} with coupling value $\lambda > 0$ if there exist action-independent kernels $\Phi_U$ and $\varphi_{U,i}$, defined by
\[
(\Phi_U)_{w,w^+} \triangleq \mathbb{P}(w^+ \mid w), \ \  \varphi_{U,i}( x_i,s_i) \triangleq \mathbb{P}(x_i^+\mid x_i,s_i),
\]
such that the action-dependent kernels $\Phi(a)$ and $\varphi_i( x_i,s_i,a_i)$ satisfy
\[
\lVert \Phi(a) - \Phi_U \rVert_{r,\infty} \le \varepsilon_\Phi, \qquad \forall a \in \mathcal{A},
\]
\[
\lVert \varphi_i( x_i,s_i,a_i) - \varphi_{U,i}( x_i,s_i) \rVert_{r,\infty} \le \varepsilon_\varphi, \ \  \forall a_i \in \mathcal{A}_i,\ i \in \mathcal{I},
\]
where $\|\cdot\|_{r,\infty}$ denotes the row-sum norm. 
Here, $\Phi(a)$ and $\Phi_U$ are treated as matrices 
row-indexed by $w$ and column-indexed by $w^+$, for each 
fixed $a$, while 
$\varphi_i(x_i, s_i, a_i)$ and $\varphi_{U,i}(x_i, s_i)$ 
are treated as matrices row-indexed by $(x_i, s_i)$ and 
column-indexed by $x_i^+$, for each fixed $a_i \in \mathcal{A}_i$. The overall coupling value is defined as
\[
\lambda = \varepsilon_\Phi +|\mathcal{I}|\varepsilon_\varphi.
\]
\end{definition}
\section{Q-Value Iteration Dynamics}

We consider Q-value iteration, defined below, as a natural model of independent, distributed learning dynamics. We prove that agents following these dynamics converge to an EEE in weakly-coupled stochastic games with a sufficiently small coupling value $\lambda$. We present the dynamics and the corresponding convergence results in what follows.

 The dynamics proceed as follows. For each agent, initialize $Q^{(0)}_i$ arbitrarily such that $\lVert Q^{(0)}_i\rVert_\infty\leq \frac{\lVert g_i\rVert_\infty}{1-\delta_i},$ where $\|\cdot\|_\infty := \sup_{z_i,x_i,a_i}|\cdot|$
denotes the sup-norm over all arguments of the function, and iterate as follows:

    \noindent 1) \textbf{Greedy policy update:}  
    \[
    \sigma^{(t)}_i(z_i,x_i)\in\arg\max_{a_i\in\mathcal A_i} Q^{(t)}_i(z_i,x_i,a_i).
    \]
    
    \noindent 2) \textbf{Consistent model update:}  
   \begin{multline*}
\mu^{(t)}_i(z_i, x_i)[s_i] = \lim_{h \to \infty} \mathbb{P}_{\mathbf{N}, \sigma^{(t)}}(s_i^h=s_i \mid \\
z_i^h = z_i,\ x_i^h = x_i)
\end{multline*}

    \noindent 3) \textbf{Q-update:}
\begin{multline*}
Q_i^{(t+1)}(z_i,x_i,a_i) = \sum_{s_i}\mu^{(t)}_i(z_i,x_i)[s_i]\,\Big[g_i(x_i,a_i,s_i)\\
+\delta_i\,\mathbb{E}_{x_i^+\sim\varphi_i}V^{(t)}_i(l_i(z_i,s_i),x_i^+)\Big]
\end{multline*} 
\noindent where $V^{(t)}_i(z_i,x_i)=\max_{a_i}Q^{(t)}_i(z_i,x_i,a_i)$. The fixed points of these dynamics, yielding $(Q^*, \mu^*, \sigma^*)$, 
are the empirical evidence equilibria by definition.

Note that when the environment and the agents' state dynamics are independent of the agents' 
actions, i.e.\ $w^+\sim \Phi_U(w)$ and $x_i^+\sim \varphi_{U,i}( x_i,s_i)$, the system becomes completely 
uncoupled and the stationary distribution of signals is independent of the collective strategy ${\sigma}$. The consistent model $\mu_i,$ therefore, corresponds to a single fixed distribution and learning through Q-value iteration then achieves convergence. However, when the state transition dynamics are action-dependent, the consistent model varies with the policy at each iteration through the policy's effect on the environment transition kernel. Thus, convergence of the dynamics is no longer guaranteed. Nevertheless, we show that under sufficiently weak coupling, convergence to an EEE is achievable. The following theorems establish this under both greedy and softmax policy updates.

\begin{definition}
For two joint strategies $\sigma, \bar\sigma$, we define the max-metric
\[
\|\sigma-\bar\sigma\|_\infty \;:=\; \max_{i}\; 
\max_{z_i,x_i,a_i}\;  
\big|\sigma_i(z_i,x_i)[a_i]-\bar\sigma_i(z_i,x_i)[a_i]\big|.
\]
For two joint Q-value functions $Q, \bar Q$, we define analogously
\[
\|Q-\bar Q\|_\infty \;:=\; \max_{i}\; 
\max_{z_i,x_i,a_i}\;   
\big|Q_i(z_i,x_i,a_i)-\bar Q_i(z_i,x_i,a_i)\big|.
\]

\end{definition}
We adopt the following assumptions as regularity conditions that ensure the continuity of  stationary distributions.

\begin{assumption}
\label{ass:regularity}
The following regularity conditions hold for any two joint strategies 
$\sigma$ and $\bar{\sigma}$:

    \noindent \textbf{Meyer sensitivity:} There exists a constant $\kappa > 0$ 
    such that $\|\pi - \bar\pi\|_\infty \le \kappa\,\|T - \bar{T}\|_{r,\infty}$, 
    where $T$ is the joint transition matrix over $(w,z,x)$ under 
    $\{l_i, \varphi_i, M_i\}_{i\in\mathcal{I}}$, $\Phi(a)$, and $\sigma$; 
    $\pi$ is the stationary distribution induced by $T$; and $\bar{T}$, 
    $\bar{\pi}$ are the corresponding quantities under $\bar{\sigma}$.
    
    \noindent \textbf{Minimal mass:} For each $i \in \mathcal{I}$, there exists 
    a constant $m_i > 0$ such that $\pi(z_i,x_i),\,\bar{\pi}(z_i,x_i) 
    \ge m_i$ for all $(z_i,x_i)$.
\end{assumption}

\begin{proposition}
\label{prop:QvalueStability}
Let agent $i$ have a bounded stage reward function, such that $|g_i(x_i,a_i,s_i)| \le G_i < \infty,$ and let $\mu_i^{(t)}$ and $\bar \mu_i^{(t)}$ be two sequences of probability models over signals $s_i$. Define

\[\varepsilon^{(t)}_{i,\mu} :=  \max_{z_i,x_i}\big\lVert \mu^{(t)}_i(z_i,x_i)-\bar\mu^{(t)}_{i}(z_i,x_i)\big\rVert_\infty, \  \varepsilon_{i,\mu} :=\max_t \varepsilon^{(t)}_{i,\mu}.\]
Let $Q_i^{(t)}$ and $\bar Q_i^{(t)}$ be the Q-value iterates under these potentially distinct models:
\begin{align*}
&Q_i^{(t+1)}(z_i, x_i, a_i) \\
&\quad= {\mathbb{E}}_{s_i\sim\mu^{(t)}_i,\ x_i^+\sim\varphi_i} \Big[ g_i(x_i, a_i, s_i) + \delta_i V_i^{(t)}(l_i(z_i,s_i), x_i^+) \Big], \\[4pt]
&\bar{Q}_i^{(t+1)}(z_i, x_i, a_i) \\
&\quad= {\mathbb{E}}_{s_i\sim\bar{\mu}^{(t)}_i,\ x_i^+\sim\varphi_i} \Big[ g_i(x_i, a_i, s_i) + \delta_i \bar{V}_i^{(t)}(l_i(z_i,s_i), x_i^+) \Big].
\end{align*}
where $V_i^{(t)}(z_i,x_i) = \max_{a_i} Q_i^{(t)}(z_i,x_i,a_i)$ and $\bar V_i^{(t)}(z_i,x_i) = \max_{a_i} \bar {Q}_i^{(t)}(z_i,x_i,a_i).$ Then,
\[
\lim_{t\rightarrow\infty}\|Q_i^{(t)} - \bar {Q}_i^{(t)}\|_\infty \le \frac{\varepsilon_{i,\mu} \lvert \mathcal{S}_i\rvert G_i}{(1-\delta_i)^2},
\] where $\|\cdot\|_\infty$ is the sup-norm over all states and actions.
\end{proposition}
\begin{theorem}

\label{ConvTheorem}
Suppose the stochastic game is weakly-coupled with value $\lambda$, and 
Assumptions~\ref{StochasticAssumption} and~\ref{ass:regularity} hold. 
Suppose further that an EEE $(\sigma^*,\mu^*)$ exists, and that $\sigma^*$ 
has a margin
\begin{align*}
\xi_i &= \min_{z_i,x_i}\Big(Q_i^*(z_i,x_i,\sigma^*_i(z_i,x_i)) - 
\max_{a''_i\neq\sigma^*_i(z_i,x_i)}Q_i^*(z_i,x_i,a_i'')\Big)\\
&> 0.
\end{align*}
for all $i \in \mathcal{I}$, where $Q^*$ is the fixed point of the Bellman equation under model $\mu^*.$ Then the Q-value iteration converges to $Q^*,$ and hence the model and policies converge to $(\sigma^*,\mu^*)$ for sufficiently small $\lambda$.
\end{theorem} 

Proposition~\ref{prop:QvalueStability} bounds the Q-value difference under bounded model perturbation, in the spirit of approximate and robust dynamic programming as in \cite{Bertsekas2012} and \cite{RoyXUPokutta2017}. Theorem~\ref{ConvTheorem} then uses this bound along with a bound on model perturbation in weakly-coupled stochastic games, governed by the coupling value, to show that the Q-values converge to a neighborhood determined by the coupling strength and game parameters, and further converge to a fixed point if the optimal action is preserved within this neighborhood.

 We further consider a relaxed optimality condition by adopting a softmax policy, and establish a sufficient condition for the convergence of the Q-value iteration. The policy update governed by a softmax policy is:
\[
\sigma^{(t)}_i(z_i,x_i)[a_i] = \text{softmax}_{\tau_i}(Q^{(t)}_i(z_i,x_i,.))
\]
\[
=\frac{\exp(Q^{(t)}_i(z_i,x_i,a_i)/\tau_i)}{\sum_{a'_i \in \mathcal{A}_i} \exp(Q^{(t)}_i(z_i,x_i,a'_i)/\tau_i)},
\]where $\tau_i > 0$
is agent
$i$'s softmax temperature parameter.
The fixed points of these dynamics are approximate variants of the EEE, defined as follows.
\begin{definition}
The pair $({\sigma}, {\mu})$
 is an \textbf{Approximate EEE} if it satisfies the same conditions as the EEE, except that the optimality condition is replaced by a softmax policy:
\[
\sigma_i(z_i,x_i)[a_i] = \frac{\exp(Q_i(z_i,x_i,a_i)/\tau_i)}{\sum_{a'_i \in \mathcal{A}_i} \exp(Q_i(z_i,x_i,a'_i)/\tau_i)}
\]
where $\tau_i > 0$
is agent
$i$'s softmax temperature parameter.
\end{definition}

\begin{theorem}
\label{prop:Contraction}
Under Assumptions~\ref{StochasticAssumption} and~\ref{ass:regularity}, 
the Q-value iteration operator with softmax updates is a contraction mapping 
on the space of Q-value functions, under the max-norm, for sufficiently small $\lambda$. 
Consequently, the agents' policies and models converge to an approximate EEE.
\end{theorem}

\section{Numerical Illustration}
To illustrate the EEE framework, we present a numerical example and 
compute the resulting EEE in this setting. We further apply the 
Q-value iteration dynamics to this example, demonstrating the effect 
of the coupling strength on the agents' behavior.

\begin{example} \label{Example}
Consider two agents in a stochastic game where the environment has $|\mathcal{W}|=4$ states, yet each agent only observes one of two individual signals, $|\mathcal{S}_1|=|\mathcal{S}_2|=2$. Each agent has two local states $|\mathcal{X}_1|=|\mathcal{X}_2|=2$, keeps in memory the last signal observed, $z^+=s$, forms a model, and chooses an action based on $(z_i,x_i)$. The environment dynamics kernel is a convex combination of an uncoupled kernel $\Phi_U$ and a coupled kernel $\Phi_C(a)$, such that $\Phi(a)=\alpha\Phi_C(a)+(1-\alpha)\Phi_U$ for all $a$. The agents state dynamics kernel is similarly a convex combination of a uniform uncoupled kernel $\varphi_U(x_i,s_i) = \frac{1}{|\mathcal{X}_i|}$ for all $x_i,s_i$, and a coupled kernel $\varphi_{C,i}(x_i,s_i,a_i)$ such that  $\varphi_{i}(x_i,s_i,a_i)=\alpha\varphi_{C,i}(x_i,s_i,a_i)+(1-\alpha)\varphi_U(x_i,s_i)$ for all $a_i$. We set $\alpha=0.9$ and $\delta_1=\delta_2=0.7$. The signal kernels, transition kernels and rewards are:
\[
M_1=\begin{pmatrix}0.98&0.02\\0.09&0.91\\0.79&0.21\\0.74&0.26\end{pmatrix},\quad
M_2=\begin{pmatrix}0.93&0.07\\0.82&0.18\\0.64&0.36\\0.11&0.89\end{pmatrix},
\]
\[\Phi_U=\begin{pmatrix}0.36&0.42&0.05&0.17\\0.06&0.42&0.33&0.19\\0.34&0.03&0.03&0.60\\0.39&0.29&0.24&0.08\end{pmatrix},\]
\[
\Phi_C(1,1)=\begin{pmatrix}0.29&0.09&0.15&0.47\\0.11&0.06&0.19&0.64\\0.25&0.29&0.21&0.25\\0.11&0.40&0.02&0.47\end{pmatrix},\]\[
\Phi_C(1,2)=\begin{pmatrix}0.06&0.51&0.20&0.23\\0.48&0.11&0.30&0.11\\0.31&0.39&0.22&0.08\\0.32&0.01&0.24&0.43\end{pmatrix},\]\[
\Phi_C(2,1)=\begin{pmatrix}0.39&0.17&0.20&0.24\\0.20&0.48&0.05&0.27\\0.09&0.48&0.30&0.13\\0.23&0.07&0.22&0.48\end{pmatrix},
\]
\[
\Phi_C(2,2)=\begin{pmatrix}0.22&0.27&0.26&0.25\\0.09&0.35&0.47&0.09\\0.38&0.27&0.22&0.13\\0.23&0.04&0.44&0.29\end{pmatrix},
\]

where $\Phi_C(a_1,a_2)$ denotes the coupled kernel under joint action $(a_1,a_2)$, with rows indexing $w$ and columns indexing $w^+$.
The local state transitions $\varphi_{C,i}(a_i)$, treated as 
matrices row-indexed by $(x_i,s_i)$ and column-indexed by 
$x_i^+$, for each $a_i$, are:
\[
\varphi_{C,1}(\cdot,\cdot,1)=\begin{pmatrix}
0.80&0.20\\
0.26&0.74\\
0.84&0.16\\
0.93&0.07
\end{pmatrix},\ 
\varphi_{C,1}(\cdot,\cdot,2)=\begin{pmatrix}
0.82&0.18\\
0.60&0.40\\
0.24&0.76\\
0.35&0.65
\end{pmatrix},
\]
\[
\varphi_{C,2}(\cdot,\cdot,1)=\begin{pmatrix}
0.34&0.66\\
0.62&0.38\\
0.64&0.36\\
0.62&0.38
\end{pmatrix},\ 
\varphi_{C,2}(\cdot,\cdot,2)=\begin{pmatrix}
0.17&0.83\\
0.61&0.39\\
0.37&0.63\\
0.48&0.52
\end{pmatrix}.
\]

 The stage rewards, identical across local states with rows indexing actions and columns indexing signals, are:
\[
g_1=\begin{pmatrix}-48&-50\\-36&1\end{pmatrix},\quad
g_2=\begin{pmatrix}31&-100\\-21&-37\end{pmatrix}.
\]
\textbf{EEE:} The pair $(\sigma^*,\mu^*)$ defined below is an EEE of this game. The optimal strategies are
\[
\sigma_1^*(z_1,x_1)=\begin{pmatrix}0&1\end{pmatrix},\ 
\sigma_2^*(z_2,x_2)=\begin{pmatrix}1&0\end{pmatrix},\  \forall\, z_1,z_2,x_1,x_2,
\]

and the consistent models $\mu_i^*(z_i,x_i)$ are
\[
\mu_1^*(z_1=1,x_1)\approx\begin{pmatrix}0.67&0.33\end{pmatrix},\quad \forall x_1
\]
\[
\mu_1^*(z_1=2,x_1)\approx\begin{pmatrix}0.54&0.46\end{pmatrix},\quad \forall x_1
\]
\[
\mu_2^*(z_2=1,x_2)\approx\begin{pmatrix}0.64&0.36\end{pmatrix},\quad \forall x_2
\]
\[
\mu_2^*(z_2=2,x_2)\approx\begin{pmatrix}0.55&0.45\end{pmatrix},\quad \forall x_2
\]
\end{example}

We now demonstrate the Q-value iteration dynamics on Example~\ref{Example} under two coupling values and the two policy update rules. In this example, $\alpha$ controls the degree of coupling, since
\[
\|\Phi(a) - \Phi_U\|_{r,\infty} = \alpha \|\Phi_C(a) - \Phi_U\|_{r,\infty},
\] 
\vspace{-5ex}
\begin{align*}
&\|\varphi_i(x_i,s_i,a_i) - \varphi_U(x_i,s_i)\|_{r,\infty} \nonumber\\ 
&\quad = \alpha \|\varphi_{C,i}(x_i,s_i,a_i) - \varphi_U(x_i,s_i)\|_{r,\infty},
\end{align*}
so that $\alpha=0$ recovers the uncoupled environment. The evolution of the agents' policies across iterations for both the greedy and softmax updates, under $\alpha=0.9$ and $\alpha=1$ is shown in Figures~\ref{fig:greedy55}--\ref{fig:soft100}.

Under the weaker coupling ($\alpha=0.9$), both the greedy and the 
softmax policies (Figures~\ref{fig:greedy55} and~\ref{fig:soft55}) converge. When $\alpha=1,$ the strategy of the second agent cycles (Figures~\ref{fig:greedy100} 
and~\ref{fig:soft100}), showing that strong coupling can disrupt convergence. In practice, 
however, convergence is commonly observed across a broad range of coupling values, suggesting 
that the Q-value iteration 
can converge beyond the proven regime.
\begin{figure}[h]
    \centering
        \centering
        \includegraphics[scale=0.33]{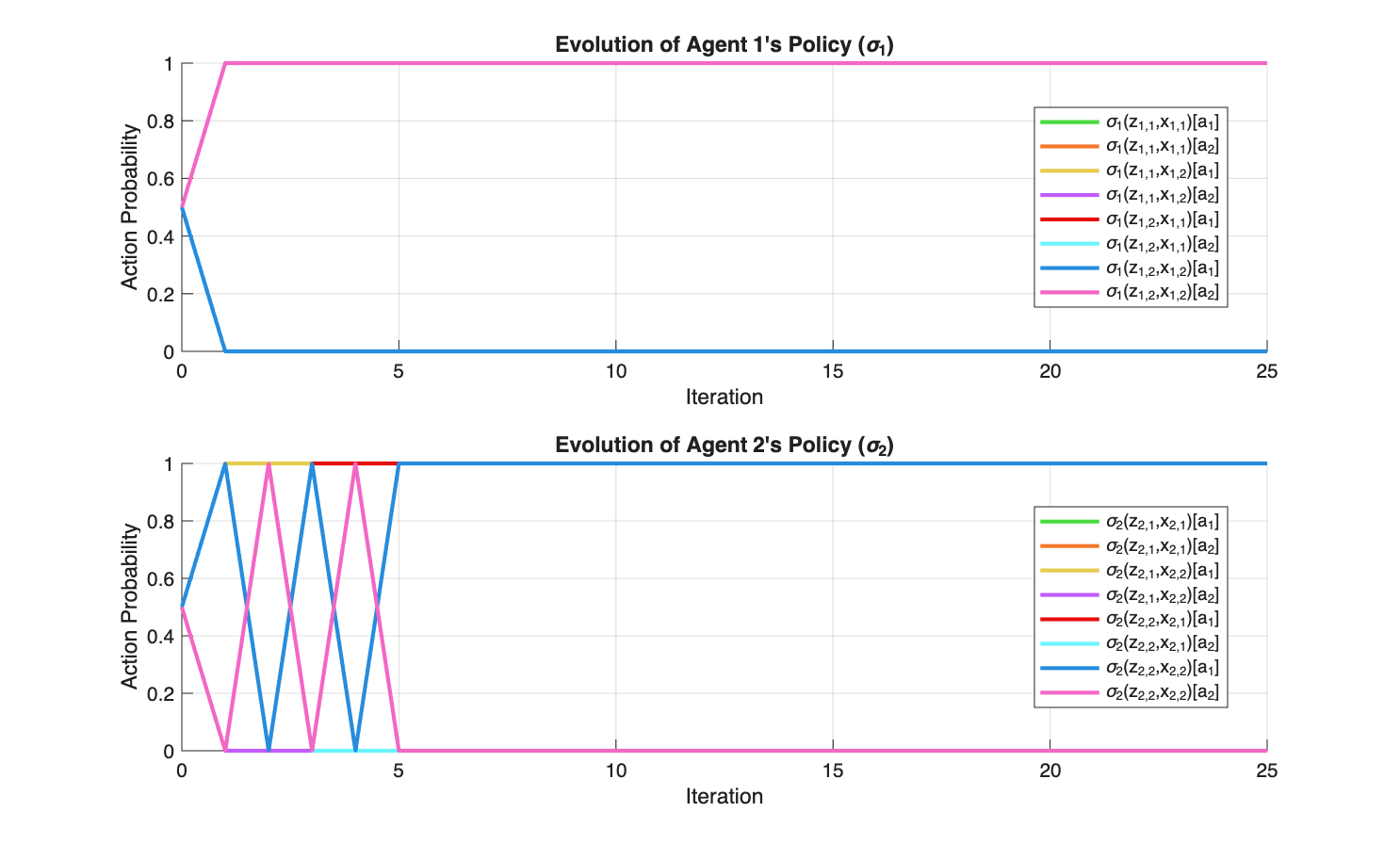}
        \caption{Greedy policy, $\alpha=0.9$}
        \label{fig:greedy55}
        \end{figure}
    \begin{figure}[h]
        \centering
        \includegraphics[scale=0.33]{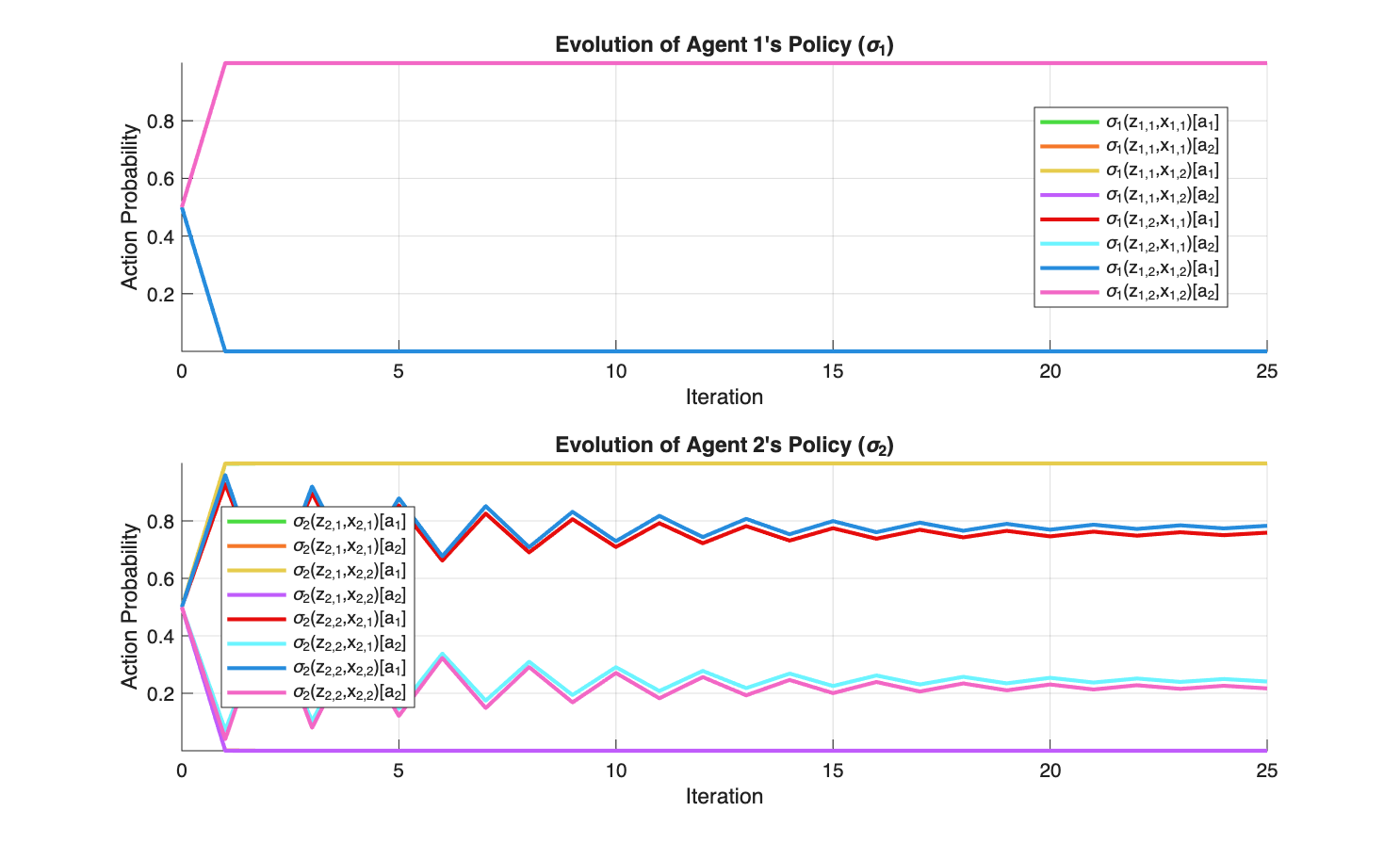}
        \caption{Softmax policy, $\tau_1=\tau_2=1,\,\alpha=0.9$}
        \label{fig:soft55}
    \end{figure}
    
    \vspace{1em}
    
    \begin{figure}
        \centering
        \includegraphics[scale=0.33]{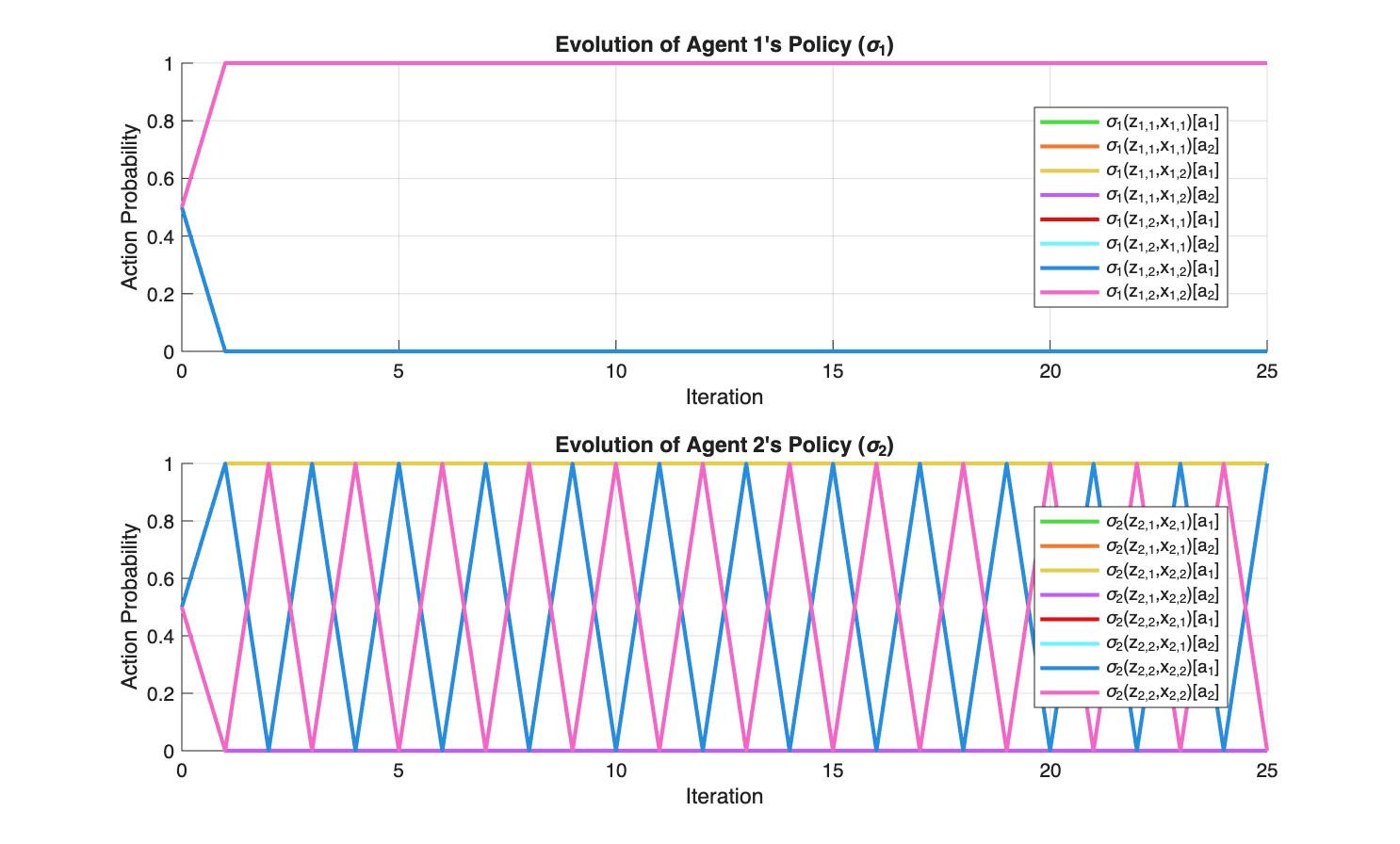}
        \caption{Greedy policy, $\alpha=1$}
        \label{fig:greedy100}
    \end{figure}
    \begin{figure}
        \centering
        \includegraphics[scale=0.33]{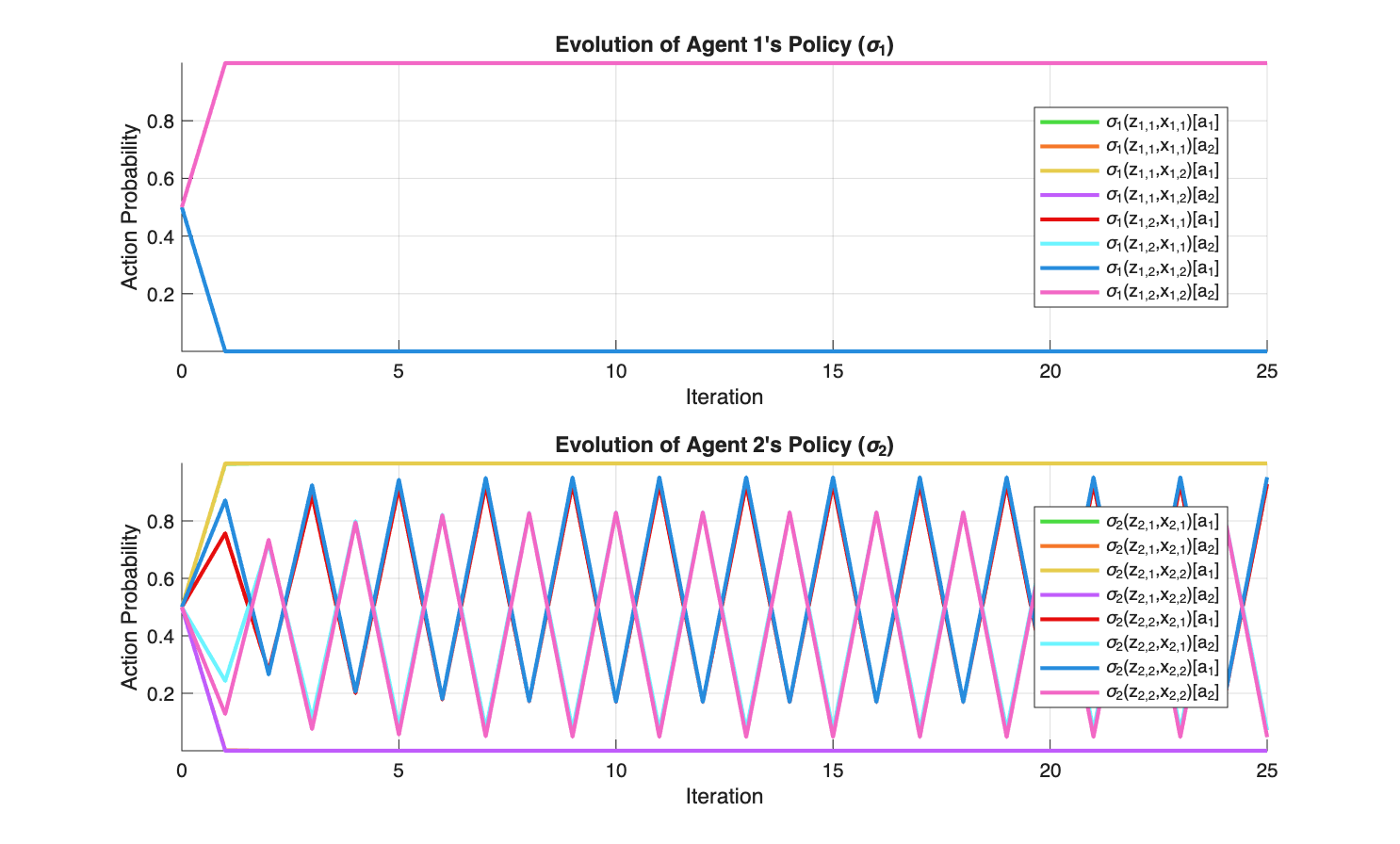}
        \caption{Softmax policy, $\tau_1=\tau_2=1,\,\alpha=1$}
        \label{fig:soft100}
    \end{figure}
\section{Conclusion}
In this work, we studied multi-agent systems operating in stochastic environments governed jointly by exogenous factors and by agents' actions, where the latter have a bounded effect on the environment's transition dynamics. Agents in this setting perceive only partial, noisy signals from the environment, forming internal models and acting upon them in a distributed fashion. We analyzed what emerges when agents implement Q-value iteration with per-iteration model updates, and proved that sufficiently small coupling between agents and the environment guarantees convergence to a steady state; specifically, that the dynamics converge to an empirical evidence equilibrium (EEE). We further extended the framework by replacing the greedy optimization step in Q-value iteration with a softmax policy, establishing a contraction result that holds under sufficiently small coupling. The fully exogenous regime is well studied, as the absence of feedback decouples each agent's problem into an independent MDP, while the present work addresses the intermediate case where agents do affect the environment, but their influence is bounded. What remains open is to characterize the behavior of the system as the environment becomes increasingly governed by agents’ actions relative to the exogenous process.

\appendix
\renewcommand\thesection{\Alph{section}}
\renewcommand\thetheorem{\thesection.\arabic{theorem}}
\setcounter{section}{1}
\setcounter{theorem}{0}
We provide proofs of all results stated in the main text, in order of appearance.
\begin{proofOf}{Proposition~\ref{prop:QvalueStability}} \textbf{(1) Difference decomposition.}
Let $\varepsilon_{i,Q}^{(t)} := \|Q_i^{(t)} - \bar{Q}_i^{(t)}\|_\infty,$  where throughout this proof, $\|\cdot\|_\infty$ denotes the sup-norm over all arguments of the function. Fix $(z_i,x_i,a_i)$. Denote $\varphi_i^{s_i}(x_i^+) := \varphi_i(x_i^+\mid x_i,s_i,a_i)$
for brevity. Then
\begin{align*}
&|Q_i^{(t+1)} - \bar{Q}_i^{(t+1)}|\\
&= \Big|\sum_{s_i}\mu_i^{(t)}[s_i]\sum_{x_i^+}
\varphi_i^{s_i}(x_i^+)
\big[g_i + \delta_i V_i^{(t)}(l_i(z_i,s_i),x_i^+)\big]\\
&\quad-\sum_{s_i}\bar\mu_i^{(t)}[s_i]\sum_{x_i^+}
\varphi_i^{s_i}(x_i^+)
\big[g_i + \delta_i \bar V_i^{(t)}(l_i(z_i,s_i),x_i^+)\big]\Big|.
\end{align*}
Adding and subtracting the cross term
\[
\sum_{s_i}\bar\mu_i^{(t)}[s_i]\sum_{x_i^+}
\varphi_i^{s_i}(x_i^+)
\big[g_i+\delta_i V_i^{(t)}(l_i(z_i,s_i),x_i^+)\big],
\]
and applying the triangle inequality:
\begin{align*}
&|Q_i^{(t+1)} - \bar{Q}_i^{(t+1)}|\\
&\leq \Big|\sum_{s_i}(\mu_i^{(t)}-\bar\mu_i^{(t)})[s_i]
\sum_{x_i^+}\varphi_i^{s_i}(x_i^+)\\
&\qquad\cdot\big[g_i+\delta_i V_i^{(t)}(l_i(z_i,s_i),x_i^+)\big]\Big|\\
&\quad+\delta_i\Big|\sum_{s_i}\bar\mu_i^{(t)}[s_i]
\sum_{x_i^+}\varphi_i^{s_i}(x_i^+)\\
&\qquad\cdot\big[V_i^{(t)}(l_i(z_i,s_i),x_i^+)
-\bar V_i^{(t)}(l_i(z_i,s_i),x_i^+)\big]\Big|\\
&\leq \sum_{s_i}|(\mu_i^{(t)}-\bar\mu_i^{(t)})[s_i]|
\sum_{x_i^+}\varphi_i^{s_i}(x_i^+)\\
&\qquad\cdot\big|g_i+\delta_i V_i^{(t)}(l_i(z_i,s_i),x_i^+)\big|\\
&\quad+\delta_i\sum_{s_i}\bar\mu_i^{(t)}[s_i]
\sum_{x_i^+}\varphi_i^{s_i}(x_i^+)\\
&\qquad\cdot\big|V_i^{(t)}(l_i(z_i,s_i),x_i^+)
-\bar V_i^{(t)}(l_i(z_i,s_i),x_i^+)\big|.
\end{align*}
For the first term, applying $|g_i + \delta_i V_i^{(t)}|
\leq |g_i| + \delta_i|V_i^{(t)}|$
and bounding by sup-norms:
\begin{align*}
&\sum_{s_i}|(\mu_i^{(t)}-\bar\mu_i^{(t)})[s_i]|
\sum_{x_i^+}\varphi_i^{s_i}(x_i^+)
\big|g_i+\delta_i V_i^{(t)}(l_i(z_i,s_i),x_i^+)\big|\\
&\leq \sum_{s_i}|(\mu_i^{(t)}-\bar\mu_i^{(t)})[s_i]|
\big(G_i + \delta_i\|V_i^{(t)}\|_\infty\big)
{\sum_{x_i^+}\varphi_i^{s_i}(x_i^+)}\\
&= \sum_{s_i}|(\mu_i^{(t)}-\bar\mu_i^{(t)})[s_i]|
\big(G_i+\delta_i\|V_i^{(t)}\|_\infty\big).
\end{align*}
For the second term, pulling $\|V_i^{(t)}-\bar V_i^{(t)}\|_\infty$
outside and using $\sum_{x_i^+}\varphi_i^{s_i}(x_i^+)=1$
and $\sum_{s_i}\bar\mu_i^{(t)}[s_i]=1$:
\begin{align*}
&\delta_i\sum_{s_i}\bar\mu_i^{(t)}[s_i]
\sum_{x_i^+}\varphi_i^{s_i}(x_i^+)\\
&\quad\cdot\big|V_i^{(t)}(l_i(z_i,s_i),x_i^+)
-\bar V_i^{(t)}(l_i(z_i,s_i),x_i^+)\big|\\
&\leq \delta_i\|V_i^{(t)}-\bar V_i^{(t)}\|_\infty
{\sum_{s_i}\bar\mu_i^{(t)}[s_i]
\sum_{x_i^+}\varphi_i^{s_i}(x_i^+)}\\
&= \delta_i\|V_i^{(t)}-\bar V_i^{(t)}\|_\infty.
\end{align*}
Combining both terms:
\begin{align*}
&|Q_i^{(t+1)} - \bar{Q}_i^{(t+1)}|\\
&\leq \sum_{s_i}|(\mu_i^{(t)}-\bar\mu_i^{(t)})[s_i]|
\big(G_i+\delta_i\|V_i^{(t)}\|_\infty\big)\\
&\quad+\delta_i\|V_i^{(t)}-\bar V_i^{(t)}\|_\infty.
\end{align*}
\textbf{(2) Bounding $\|V_i^{(t)}\|_\infty$ and 
$\|V_i^{(t)} - \bar{V}_i^{(t)}\|_\infty$.}
We prove by induction that 
$\|Q_i^{(t)}\|_\infty \le \frac{G_i}{1-\delta_i}$ 
and $\|V_i^{(t)}\|_\infty \le \frac{G_i}{1-\delta_i}.$

\textit{Base case:} By initialization, 
$\|Q_i^{(0)}\|_\infty \le \frac{G_i}{1-\delta_i}$, 
hence $\|V_i^{(0)}\|_\infty = \|\max_{a_i} Q_i^{(0)}\|_\infty 
\le \frac{G_i}{1-\delta_i}.$

\textit{Inductive step:} Suppose 
$\|Q_i^{(t)}\|_\infty \le \frac{G_i}{1-\delta_i}$, 
so $\|V_i^{(t)}\|_\infty \le \frac{G_i}{1-\delta_i}.$
Then for any $(z_i,x_i,a_i)$:
\begin{align*}
&|Q_i^{(t+1)}(z_i,x_i,a_i)|\\
&= \Big|\sum_{s_i}\mu_i^{(t)}[s_i]\sum_{x_i^+}
\varphi_i^{s_i}(x_i^+)
\big[g_i + \delta_i V_i^{(t)}(l_i(z_i,s_i),x_i^+)\big]\Big|\\
&\leq \sum_{s_i}\mu_i^{(t)}[s_i]\sum_{x_i^+}
\varphi_i^{s_i}(x_i^+)
\big(G_i + \delta_i\|V_i^{(t)}\|_\infty\big)\\
&= G_i + \delta_i\|V_i^{(t)}\|_\infty\\
&\leq G_i + \delta_i\frac{G_i}{1-\delta_i}
= \frac{G_i}{1-\delta_i},
\end{align*}
where we used $\sum_{s_i}\mu_i^{(t)}[s_i]=1$ 
and $\sum_{x_i^+}\varphi_i^{s_i}(x_i^+)=1.$

Now, examine $\rvert V_i^{(t)} - \bar{V}_i^{(t)} \lvert$ for a fixed $z_i$ and $x_i.$
\begin{align*}
&|V_i^{(t)}(z_i,x_i) - \bar{V}_i^{(t)}(z_i,x_i)|\\
&= |\max_{a_i}Q_i^{(t)}(z_i,x_i,a_i) 
- \max_{a_i}\bar{Q}_i^{(t)}(z_i,x_i,a_i)|\\
&\leq \max_{a_i}|Q_i^{(t)}(z_i,x_i,a_i) 
- \bar{Q}_i^{(t)}(z_i,x_i,a_i)|\\
&\leq \varepsilon_{i,Q}^{(t)},
\end{align*}
hence $\|V_i^{(t)} - \bar{V}_i^{(t)}\|_\infty 
\leq \varepsilon_{i,Q}^{(t)}.$

\textbf{(3) Bounding $\varepsilon_{i,Q}^{(t+1)}$.}
Substituting the bounds from Steps~(1) and~(2), 
taking the sup over $(z_i,x_i,a_i)$, and using
\[
\sum_{s_i}|\mu_i^{(t)}[s_i]-\bar\mu_i^{(t)}[s_i]|
\leq |\mathcal{S}_i|\,\varepsilon_{i,\mu}^{(t)},
\]
we obtain:
\begin{align*}
\varepsilon_{i,Q}^{(t+1)}
&\le \varepsilon^{(t)}_{i,\mu}|\mathcal{S}_i|G_i  
+ \delta_i\varepsilon^{(t)}_{i,\mu}|\mathcal{S}_i|
\frac{G_i}{1-\delta_i} + \delta_i\varepsilon_{i,Q}^{(t)}\\
&= \frac{\varepsilon^{(t)}_{i,\mu}|\mathcal{S}_i|G_i}{1-\delta_i} 
+ \delta_i\varepsilon_{i,Q}^{(t)}\\
&\le \frac{\varepsilon_{i,\mu}|\mathcal{S}_i|G_i}{1-\delta_i} 
+ \delta_i\varepsilon_{i,Q}^{(t)}.
\end{align*}
\textbf{(4) Solving the recursion.} Iterating gives
\[
\varepsilon_{i,Q}^{(t+1)} \le 
\frac{\varepsilon_{i,\mu}|\mathcal{S}_i|G_i}{1-\delta_i}
\sum_{k=0}^{t}\delta_i^k + \delta_i^{t+1}\varepsilon_{i,Q}^{(0)}
\le \frac{\varepsilon_{i,\mu}|\mathcal{S}_i|G_i}{(1-\delta_i)^2} 
+ \delta_i^{t+1}\varepsilon_{i,Q}^{(0)}.
\]
Taking $t\to\infty$ and using $\delta_i\in(0,1)$:
\[
\lim_{t\to\infty}\|Q_i^{(t)} - \bar Q_i^{(t)}\|_\infty 
\le \frac{\varepsilon_{i,\mu}|\mathcal{S}_i|G_i}{(1-\delta_i)^2}.
\]

\end{proofOf}
\begin{proposition}
\label{prop:EnvToModel}
Consider a weakly-coupled stochastic game with coupling 
value $\lambda.$ Let $\sigma$ and $\bar\sigma$ 
be two distinct joint strategies, and let $\mu$ and $\bar\mu$ 
be the corresponding consistent agent models. Let $\pi$ and 
$\bar\pi$ denote the stationary distributions of the joint 
state process $(w,z,x)$ induced by $\mathbf{N}$ and $\sigma$, 
and by $\mathbf{N}$ and $\bar\sigma$, respectively, whose 
existence and uniqueness are guaranteed by 
Assumption~\ref{StochasticAssumption}. Let $\kappa$, $m_i$ 
be the constants from Assumption~\ref{ass:regularity}.
Under Assumptions~\ref{StochasticAssumption} 
and~\ref{ass:regularity}, for any agent $i$, 
the following bound holds:
\begin{align*}
\|\mu_i(z_i,x_i)-\bar{\mu}_{i}(z_i,x_i)\|_\infty
&\le
\frac{(1+p_i^{\max})\,E_i(z_i,x_i)}{m_i}  \\
&\hspace{-7em} \le \frac{(1+p_i^{\max})\kappa\,|\mathcal W|\,
|\mathcal Z_{-i}|\,|\mathcal X_{-i}| \, 
\hat{\mathcal{A}} \|\sigma - \bar\sigma\|_\infty \lambda}{m_i},
\end{align*} 
where $E_i(z_i,x_i)
= \sum_{w,z_{-i},x_{-i}}
|\pi(w,z,x)-\bar{\pi}(w,z,x)|$, 
$\hat{\mathcal{A}} = \sum_{j=1}^{|\mathcal{I}|} |\mathcal{A}_j|$, 
and $p_i^{\max}=\max_{s_i,w} \mathbb{P}(s_i\mid w).$
\end{proposition}

\begin{proof} \textbf{(1) Bounding the perturbation of the joint transition matrix.} 
Let the current joint state be $\psi := (w,z,x)$ and the 
next joint state $\psi^+ := (w^+,z^+,x^+)$. The transition 
matrix difference at $(\psi,\psi^+)$ is:
\begin{multline*}
|T_{\psi,\psi^+} - \bar{T}_{\psi,\psi^+}| \\
= \left|\sum_{a}\prod_{i\in\mathcal{I}}
\sum_{s_i}\mathbb{P}(s_i\mid w)\,
\mathbf{1}_{z_i^+=l_i(z_i,s_i)}\,
\varphi_i(x_i,s_i,a_i)[x_i^+]
\,\right.\\
\left.\cdot\Phi_{ww^+}(a)\left(\sigma(z,x)[a]
-\bar{\sigma}(z,x)[a]\right)\right|.
\end{multline*}
Since each agent's policy $\sigma_i$ is local, 
$\sigma(z,x)[a] = \prod_{i\in\mathcal{I}}\sigma_i(z_i,x_i)[a_i]$.
Define $\varepsilon_{\sigma} := \|\sigma-\bar{\sigma}\|_\infty.$
By the weak coupling of the game with value $\lambda$ 
(Definition~\ref{weak}), there exist action-independent kernels 
$\varphi_{U,i}(x_i,s_i)$, $\Phi_U$, and constants 
$\varepsilon_\Phi,\varepsilon_\varphi>0$ such that
\[
\|\varphi_i(x_i,s_i,a_i)-\varphi_{U,i}(x_i,s_i)\|_{r,\infty}
\le\varepsilon_\varphi, \ 
\|\Phi(a)-\Phi_U\|_{r,\infty}\le\varepsilon_\Phi,
\]
for all $a\in\mathcal{A}$ and $i\in\mathcal{I}$. Define the 
marginal transition kernels
\[
\mathbb{P}(x_i^+\mid x_i,a_i,w) := \sum_{s_i}\mathbb{P}(s_i\mid w)\,
\varphi_i(x_i,s_i,a_i)[x_i^+],
\]
\[
\mathbb{P}_U(x_i^+\mid x_i,w) := \sum_{s_i}\mathbb{P}(s_i\mid w)\,
\varphi_{U,i}(x_i,s_i)[x_i^+],
\]
where $\mathbb{P}(\cdot\mid x_i,a_i,w)$ is action-dependent 
and $\mathbb{P}_U(\cdot\mid x_i,w)$ is its action-independent 
counterpart corresponding to $\varphi_{U,i}$, satisfying

\[\sum_{x_i^+}\left|\mathbb{P}(x_i^+\mid x_i,a_i,w)
-\mathbb{P}_U(x_i^+\mid x_i,w)\right|
\le\varepsilon_\varphi, \quad \forall\, a_i\in\mathcal{A}_i.\]
Since $\sigma$ and $\bar\sigma$ are probability distributions, 
$\sum_a(\sigma(z,x)[a]-\bar\sigma(z,x)[a])=0$, and as 
$\mathbb{P}_U$ and $\Phi_U$ are action-independent:
\begin{align*}
&\sum_a\prod_{i\in\mathcal{I}}\mathbb{P}_U(x_i^+\mid x_i,w)
\cdot(\Phi_U)_{ww^+}
\left(\sigma(z,x)[a]-\bar\sigma(z,x)[a]\right)\\
&=0.
\end{align*}
Hence, we subtract this term inside the absolute value 
without changing its value. Since $z_i^+=l_i(z_i,s_i)$ 
is a deterministic function of $(z_i,s_i)$, the indicator 
$\mathbf{1}_{z_i^+=l_i(z_i,s_i)}$ fixes $z^+$ given 
$(z,s)$, so the sum over $\psi^+=(w^+,z^+,x^+)$ reduces 
to a sum over $(w^+,x^+)$ only. We then split the 
row-sum difference into Term~1 and Term~2:
\begin{align*}
&\sum_{\psi^+}|T_{\psi,\psi^+}-\bar{T}_{\psi,\psi^+}|\\
&=\sum_{x^+,w^+}\left|\sum_{a}\left(
\prod_{i\in\mathcal{I}}\mathbb{P}(x_i^+\mid x_i,a_i,w)
\cdot\Phi_{ww^+}(a)-\right.\right.\\
&\quad\left.\left.
\prod_{i\in\mathcal{I}}\mathbb{P}_U(x_i^+\mid x_i,w)
\cdot(\Phi_U)_{ww^+}\right)
\left(\sigma(z,x)[a]-\bar{\sigma}(z,x)[a]\right)\right|\\
&\leq \sum_{x^+,w^+}\left|\sum_a\text{Term 1}
\cdot(\sigma-\bar\sigma)\right|\\
&\quad+\sum_{x^+,w^+}\left|\sum_a\text{Term 2}
\cdot(\sigma-\bar\sigma)\right|,
\end{align*}
where
\begin{align*}
\text{Term 1}&:=\left(\prod_{i\in\mathcal{I}}
\mathbb{P}(x_i^+\mid x_i,a_i,w)
-\prod_{i\in\mathcal{I}}\mathbb{P}_U(x_i^+\mid x_i,w)
\right)\\
&\quad\cdot\Phi_{ww^+}(a),\\
\text{Term 2}&:=\prod_{i\in\mathcal{I}}
\mathbb{P}_U(x_i^+\mid x_i,w)
\cdot\left(\Phi_{ww^+}(a)-(\Phi_U)_{ww^+}\right).
\end{align*}

We will make use of the following telescoping argument 
throughout the proof. For two sequences of scalars $p_i, \bar{p}_i \in [0,1]$, 
by successively adding and subtracting intermediate terms:
\begin{align*}
&\prod_{i\in\mathcal{I}} p_i 
- \prod_{i\in\mathcal{I}} \bar{p}_i \\
&= p_1\prod_{i\geq 2}p_i
{-\bar{p}_1\prod_{i\geq 2}p_i
+\bar{p}_1\prod_{i\geq 2}p_i}
-\bar{p}_1\bar{p}_2\prod_{i\geq 3}p_i
{+\cdots}
-\prod_{i\in\mathcal{I}}\bar{p}_i\\
&=(p_1-\bar{p}_1)\prod_{i\geq 2}p_i
+\bar{p}_1(p_2-\bar{p}_2)\prod_{i\geq 3}p_i
+\cdots\\
&\quad+\left(\prod_{i<|\mathcal{I}|}\bar{p}_i\right)
(p_{|\mathcal{I}|}-\bar{p}_{|\mathcal{I}|})\\
&= \sum_{j\in\mathcal{I}} \left(\prod_{i<j}\bar{p}_i\right)
(p_j - \bar{p}_j)
\left(\prod_{i>j}p_i\right).
\end{align*}

Taking absolute values and using $p_i, \bar{p}_i \le 1$:
\[
\left|\prod_{i\in\mathcal{I}} p_i 
- \prod_{i\in\mathcal{I}} \bar{p}_i\right|
\leq \sum_{j\in\mathcal{I}} 
\left|p_j - \bar{p}_j\right|
\leq |\mathcal{I}|\,\varepsilon,
\]
whenever $|p_j - \bar{p}_j| \le \varepsilon$ for all 
$j\in\mathcal{I}$. Furthermore, if $p_i = 
\sigma_i(z_i,x_i)[a_i]$ and $\bar{p}_i = 
\bar{\sigma}_i(z_i,x_i)[a_i]$ are probability distributions, 
then $\prod_{i<j}\bar{p}_i\prod_{i>j}p_i$ sum to $1$ 
over $a_{-j}$, giving:
\begin{align*}
&\sum_a |\sigma(z,x)[a] - \bar{\sigma}(z,x)[a]| \\
&\leq \sum_{j\in\mathcal{I}} \sum_{a_j} 
\left|\sigma_j(z_j,x_j)[a_j] 
- \bar{\sigma}_j(z_j,x_j)[a_j]\right|\\
&\leq \sum_{j\in\mathcal{I}}|\mathcal{A}_j|
\cdot\varepsilon_\sigma
= \hat{\mathcal{A}}\,\varepsilon_\sigma.
\end{align*}
\textit{Bounding Term 2.} Since 
$\prod_{i\in\mathcal{I}}\mathbb{P}_U(x_i^+\mid x_i,w)\geq 0$ 
and 
$\sum_{x^+}\prod_{i\in\mathcal{I}}
\mathbb{P}_U(x_i^+\mid x_i,w)=1,$ 
we get:
\begin{align*}
&\sum_{x^+,w^+}\left|\sum_a \prod_{i\in\mathcal{I}}
\mathbb{P}_U(x_i^+\mid x_i,w)\right.\\
&\qquad\left.\cdot\left(\Phi_{ww^+}(a)
-(\Phi_U)_{ww^+}\right)(\sigma-\bar\sigma)\right|\\
&= \sum_{w^+}\left|\sum_a\left(\Phi_{ww^+}(a)
-(\Phi_U)_{ww^+}\right)(\sigma-\bar\sigma)\right|\\
&\leq \sum_{w^+}\sum_a
\left|\Phi_{ww^+}(a)-(\Phi_U)_{ww^+}\right|
\cdot\left|\sigma(z,x)[a]-\bar\sigma(z,x)[a]\right|\\
&\leq \left(\sum_{w^+}\max_a
\left|\Phi_{ww^+}(a)-(\Phi_U)_{ww^+}\right|\right)\\
&\qquad\cdot\sum_a\left|\sigma(z,x)[a]
-\bar\sigma(z,x)[a]\right|\\
&\leq \varepsilon_\Phi\,\hat{\mathcal{A}}\,\varepsilon_\sigma,
\end{align*}
where the first equality uses that 
$\prod_{i\in\mathcal{I}}\mathbb{P}_U(x_i^+\mid x_i,w)$ 
is nonneg\-ative and sums to $1$ over $x^+$, so it factors 
out of the absolute value; the second inequality is the 
triangle inequality; the third separates the $a$-dependence; 
and the last uses 
$\sum_{w^+}|\Phi_{ww^+}(a)-(\Phi_U)_{ww^+}|
\leq\|\Phi(a)-\Phi_U\|_{r,\infty}\leq\varepsilon_\Phi$ 
and the telescoping bound 
$\sum_a|\sigma(z,x)[a]-\bar\sigma(z,x)[a]|
\leq\hat{\mathcal{A}}\,\varepsilon_\sigma$.

\textit{Bounding Term 1.} Applying the telescoping argument 
with $p_i=\mathbb{P}(x_i^+\mid x_i,a_i,w)$ and 
$\bar{p}_i=\mathbb{P}_U(x_i^+\mid x_i,w)$, and summing 
over $x^+=(x_1^+,\ldots,x_{|\mathcal{I}|}^+)$:
\begin{align*}
&\sum_{x^+}\left|\prod_{i\in\mathcal{I}}
\mathbb{P}(x_i^+\mid x_i,a_i,w) 
- \prod_{i\in\mathcal{I}}\mathbb{P}_U(x_i^+\mid x_i,w)
\right|\\
&\leq \sum_j\sum_{x^+}
\prod_{i<j}\mathbb{P}_U(x_i^+\mid x_i,w)\\
&\quad\cdot\left|\mathbb{P}(x_j^+\mid x_j,a_j,w)
-\mathbb{P}_U(x_j^+\mid x_j,w)\right|\\
&\quad\cdot\prod_{i>j}\mathbb{P}(x_i^+\mid x_i,a_i,w)\\
&= \sum_j\sum_{x_j^+}
\left|\mathbb{P}(x_j^+\mid x_j,a_j,w)
-\mathbb{P}_U(x_j^+\mid x_j,w)\right|\\
&\leq |\mathcal{I}|\,\varepsilon_\varphi,
\end{align*}
where the equality uses that $x^+$ factorizes over agents, 
so summing over $x_i^+$ for $i\neq j$ collapses each 
marginal $\mathbb{P}_U$ and $\mathbb{P}$ to $1$, and the 
last inequality applies the row-sum norm bound 
$\sum_{x_j^+}|\mathbb{P}(x_j^+\mid x_j,a_j,w)
-\mathbb{P}_U(x_j^+\mid x_j,w)|
\leq\varepsilon_\varphi$ for each $j$.

Using this, $\Phi_{ww^+}(a)\leq 1$, and the triangle 
inequality:
\begin{align*}
&\sum_{x^+,w^+}\left|\sum_a \text{Term 1}
\cdot(\sigma-\bar\sigma)\right|\\
&\leq \sum_{x^+,w^+}\sum_a
\left|\prod_{i}\mathbb{P}(x_i^+\mid x_i,a_i,w)
-\prod_{i}\mathbb{P}_U(x_i^+\mid x_i,w)\right|\\
&\qquad\cdot\Phi_{ww^+}(a)
\cdot\left|\sigma(z,x)[a]-\bar\sigma(z,x)[a]\right|\\
&\leq |\mathcal{I}|\,\varepsilon_\varphi
\cdot\sum_{w^+}\sum_a
\left|\sigma(z,x)[a]-\bar\sigma(z,x)[a]\right|\\
&\qquad\cdot\Phi_{ww^+}(a)\\
&= |\mathcal{I}|\,\varepsilon_\varphi
\cdot\sum_a
\left|\sigma(z,x)[a]-\bar\sigma(z,x)[a]\right|\\
&\qquad\cdot{\sum_{w^+}
\Phi_{ww^+}(a)}\\
&= |\mathcal{I}|\,\varepsilon_\varphi\,
\hat{\mathcal{A}}\,\varepsilon_\sigma,
\end{align*}
where the second inequality uses 
$\sum_{x^+}|\prod_{i}\mathbb{P}-\prod_{i}\mathbb{P}_U|
\leq|\mathcal{I}|\,\varepsilon_\varphi$, the first equality 
swaps the order of summation and uses 
$\sum_{w^+}\Phi_{ww^+}(a)=1$ for each fixed $a$ and $w$ 
since $\Phi(a)$ is a stochastic matrix, and the last 
equality applies the telescoping bound 
$\sum_a|\sigma(z,x)[a]-\bar\sigma(z,x)[a]|
\leq\hat{\mathcal{A}}\,\varepsilon_\sigma$.

\textit{Combining.} Adding Term 1 and Term 2:
\begin{align*}
\sum_{\psi^+}|T_{\psi,\psi^+}-\bar{T}_{\psi,\psi^+}| 
&\le \hat{\mathcal{A}}\,\varepsilon_\sigma
\left(\varepsilon_\Phi + |\mathcal{I}|
\varepsilon_\varphi\right)
= \hat{\mathcal{A}}\,\varepsilon_\sigma\,\lambda,
\end{align*}
where $\lambda= \varepsilon_\Phi 
+ |\mathcal{I}|\varepsilon_\varphi.$

\noindent
\textbf{(2) Bounding stationary distribution differences using Meyer’s condition number.} By Meyer\cite{Meyer} and Assumption~\ref{ass:regularity}, \[ \|\pi-\bar\pi\|_\infty \le \kappa\,\|T-\bar{T}\|_{r,\infty} \le \kappa\,\hat{\mathcal{A}}\,\varepsilon_{\sigma}\,\lambda, \] where $\kappa$ is a condition number depending only on $T_U,$ where $T_U$ is the transition matrix under the uncoupled kernels $\Phi_U$ and $\varphi_{U,i}$. This implies that for every joint-state $\psi$, we have  
\[
|\pi(\psi)-\bar{\pi}(\psi)| \leq \kappa\hat{\mathcal{A}}\,\varepsilon_{\sigma}\,\lambda.
\]

\textbf{(3) Marginal and joint-probability perturbations.} Let $D(z_i,x_i)=\pi(z_i,x_i)$ and $N(s_i,z_i,x_i)=\mathbb{P}_{\pi}(s_i,z_i,x_i)=\sum_{w,\, z_{-i},\, x_{-i}}\mathbb{P}(s_i\mid w)\,\pi(w,z,x)$. The following holds: 
  \[ |D(z_i,x_i)-\bar{D}(z_i,x_i)| = \left|\sum_{w,z_{-i},x_{-i}}\big(\pi(\psi)-\bar{\pi}(\psi)\big)\right| \] \[ \le \sum_{w,\, z_{-i},\, x_{-i}} |\pi(\psi)-\bar\pi(\psi)| = E_i(z_i,x_i).\] 
\[
|N(s_i,z_i,x_i)-\bar{N}(s_i,z_i,x_i)| \]
\[= \left|\sum_{w,\, z_{-i},\, x_{-i}}\mathbb{P}(s_i\mid w)\big(\pi(\psi)-\bar\pi(\psi)\big)\right| \]
\[\le \sum_{w,\, z_{-i},\, x_{-i}} \mathbb{P}(s_i\mid w)\,|\pi(\psi)-\bar\pi(\psi)| \]
\[\le p_i^{\max} \sum_{w,\, z_{-i},\, x_{-i}} |\pi(\psi)-\bar\pi(\psi)|= p_i^{\max} E_i(z_i,x_i).
\]

\textbf{(4) Algebra for the difference of conditionals.} Using \[ \frac{N}{D}-\frac{\bar{N}}{\bar{D}}=\frac{(N-\bar{N})\bar{D} - \bar{N}(D-\bar{D})}{D \bar{D}}\] and taking absolute values, we get:

\[\hspace{-16ex}
\Big|\frac{N}{D}-\frac{\bar{N}}{\bar{D}}\Big| 
\le \frac{|N-\bar{N}|\,\bar{D} + \bar{N}\,|D-\bar{D}|}{D \bar{D}} \]
\[\le \frac{|N-\bar{N}|\,\bar{D} + \bar{D}\,|D-\bar{D}|}{D \bar{D}}\le \frac{|N-\bar{N}| + |D-\bar{D}|}{m_i}.\]
These inequalities hold since 
\begin{align*}
    \bar{N}(s_i,z_i,x_i) &=\sum_{w,\, z_{-i},\, x_{-i}}\mathbb{P}(s_i\mid w)\bar\pi(w,z,x) \\
    &\leq \sum_{w,z_{-i},x_{-i}}\bar\pi(w,z,x)=\bar{D}(z_i,x_i),
\end{align*} and $m_i\leq \pi(z_i,x_i)=D(z_i,x_i)$ by assumption.
Substituting the bounds from step (3) yields a bound for any $s_i$: 

    \[\big|\mu_i(z_i,x_i)[s_i]-\bar\mu_{i}(z_i,x_i)[s_i]\big| =\Big|\frac{N}{D}-\frac{\bar{N}}{\bar{D}}\Big|
    \]\[ \le \frac{p_i^{\max} E_i(z_i,x_i) + E_i(z_i,x_i)}{m_i} 
    = \frac{(1+p_i^{\max})E_i(z_i,x_i)}{m_i}.\]
Since this holds for any $s_i$, it holds for the infinity norm. 

Finally, using the bound in point (2) along with the definition of $E_i(z_i,x_i),$ a crude count gives us $E_i(z_i,x_i)\le |\mathcal W|\,|\mathcal Z_{-i}|\,|\mathcal X_{-i}|\,\kappa\hat{\mathcal{A}}\,\varepsilon_{\sigma}\,\lambda,$ for all $(z_i,x_i),$ which completes the proof. 
\end{proof}

\begin{proofOf}{Theorem~\ref{ConvTheorem}}
Since $Q_i^*$ is the fixed point of the Bellman equation under $\mu_i^*$, 
setting $\bar{Q}_i^{(t)} = Q_i^*$ and $\bar{\mu}_i^{(t)} = \mu_i^*$ in 
Proposition~\ref{prop:QvalueStability}, and applying the model bound from 
Proposition~\ref{prop:EnvToModel}, we have:
\vspace{-1em}
\begin{align*}
\|Q_i^{(t+1)} -  Q_i^{*}\|_\infty 
&\le \max_{z_i,x_i,k\in\{1,...,t\}}\hspace{-1em} \frac{(1+p_i^{\max})E^{(k)}_i(z_i,x_i) 
\lvert \mathcal{S}_i\rvert G_i}{m_i(1-\delta_i)^2}\\
&\quad + \delta_i^{t+1} \|Q_i^{(0)} -  Q_i^{*}\|_\infty.
\end{align*}
From Proposition~\ref{prop:EnvToModel}, $E^{(t)}_i(z_i,x_i)$ admits a bound linear in $\lambda$, and thus the first term can be arbitrarily small if $\lambda$ is sufficiently small. The second term vanishes as $t\to\infty$ since $\delta_i\in(0,1)$. Hence, for sufficiently small $\lambda$, there exists some iteration $T$ such that for all $t\geq T$ and $i\in\mathcal{I}$:
\[ \|Q_i^{(t)} - Q_i^{*}\|_\infty < \xi_i/2. \]
This implies the greedy action at $Q_i^{(t)}$, for all $t\geq T$, coincides with that of $Q_i^*$. To see this, note that for any rival action $a_i'' \neq \sigma^*_i(z_i,x_i)$:
\begin{align*}
&Q_i^{(t)}(z_i,x_i,\sigma^*_i) - Q_i^{(t)}(z_i,x_i,a_i'') \\
&= \bigl[Q_i^*(z_i,x_i,\sigma^*_i) - Q_i^*(z_i,x_i,a_i'')\bigr] \\
&\quad + \bigl[Q_i^{(t)}(z_i,x_i,\sigma^*_i) - Q_i^*(z_i,x_i,\sigma^*_i)\bigr] \\
&\quad + \bigl[Q_i^*(z_i,x_i,a_i'') - Q_i^{(t)}(z_i,x_i,a_i'')\bigr] \\
&\geq \xi_i - 2\|Q_i^{(t)} - Q_i^*\|_\infty > 0,
\end{align*}
where the last inequality holds since $\|Q_i^{(t)} - Q_i^*\|_\infty < \xi_i/2$. Hence the greedy action under $Q_i^{(t)}$ coincides with $\sigma_i^*$. Thus, the $\|\sigma^{(t)} - \sigma^*\|_\infty$ factor in $E^{(t)}_i(z_i,x_i)$ vanishes. Consequently, $\lVert \mu_i^{(t)}-\mu_i^{*} \rVert_\infty=0$ and
\[\|Q_i^{(t+1)} -  Q_i^{*}\|_\infty 
\le \delta_i \|Q_i^{(t)} -  Q_i^{*}\|_\infty \quad \forall t>T.\]
Convergence is guaranteed since 
\[\|Q_i^{(t+1)} -  Q_i^{*}\|_\infty \leq \delta_i^{t+1-T}\|Q_i^{(T)} -  Q_i^{*}\|_\infty< \delta_i^{t+1-T} \xi_i/2.\]
\end{proofOf}

\begin{proofOf}{Theorem~\ref{prop:Contraction}}
From \cite{GP18}, the softmax function with temperature $\tau_i$ is $\frac{1}{\tau_i}$-Lipschitz continuous with respect to the $2$-norm. Combining this with the standard inequalities relating the infinity and $2$-norms for vectors in $\mathbb{R}^{|\mathcal{A}_i|}$, we get:
\begin{align*}
\lVert {\sigma}_i(z_i,x_i) - \bar\sigma_i(z_i,x_i) \rVert_\infty
&\leq \lVert {\sigma}_i(z_i,x_i) - \bar\sigma_i(z_i,x_i) \rVert_2 \\
&\hspace{-5ex}\leq \frac{1}{\tau_i} \lVert {{Q}}_i(z_i,x_i) - \bar{Q}_i(z_i,x_i) \rVert_2 \\
&\hspace{-5ex}\leq \frac{\sqrt{|\mathcal{A}_i|}}{\tau_i} \lVert {{Q}}_i(z_i,x_i) - \bar{Q}_i(z_i,x_i) \rVert_\infty.
\end{align*}
Hence, 
\[
\lVert {\sigma} - \bar\sigma \rVert_\infty \leq \max_i \frac{\sqrt{|\mathcal{A}_i|}}{\tau_i} \lVert {{Q}} - \bar{Q} \rVert_\infty.
\]
From Proposition~\ref{prop:EnvToModel}, \begin{align*}
\|\mu_i(z_i,x_i)-\bar{\mu}_{i}(z_i,x_i)\|_\infty
  \\
&\hspace{-7em} \le \frac{(1+p_i^{\max})\kappa\,|\mathcal W|\,|\mathcal Z_{-i}|\,|\mathcal X_{-i}| \, \hat{\mathcal{A}} \|\sigma - \bar\sigma\|_\infty \lambda}{m_i}.
\end{align*}

Additionally, the proof of Proposition~\ref{prop:QvalueStability} shows that
\begin{align*}
\varepsilon_{i,Q}^{(t+1)}
&\le \frac{\varepsilon^{(t)}_{i,\mu} \lvert \mathcal{S}_i\rvert G_i}{1-\delta_i} + \delta_i \varepsilon_{i,Q}^{(t)},
\end{align*} where $\varepsilon_{i,Q}^{(t)} = \|Q_i^{(t)} - \bar{Q}_i^{(t)}\|_\infty$ and $\varepsilon^{(t)}_{i,\mu} =  \max_{z_i,x_i}\big\lVert \mu^{(t)}_i(z_i,x_i)-\bar\mu^{(t)}_{i}(z_i,x_i)\big\rVert_\infty.$

\vspace{1em}
Substituting the bound on $\|\sigma - \bar\sigma\|_\infty$ 
into Proposition~\ref{prop:EnvToModel}, and then 
substituting the resulting bound on 
$\varepsilon^{(t)}_{i,\mu}$ into the recursion from 
Proposition~\ref{prop:QvalueStability}:
\begin{align*}
&\varepsilon_{i,Q}^{(t+1)}\\
&\le \frac{\varepsilon^{(t)}_{i,\mu} 
\lvert \mathcal{S}_i\rvert G_i}{1-\delta_i} 
+ \delta_i \varepsilon_{i,Q}^{(t)}\\
 \\
&\le \frac{(1+p_i^{\max})\kappa\,|\mathcal{W}|\,
|\mathcal{Z}_{-i}|\,|\mathcal{X}_{-i}|\,
\hat{\mathcal{A}}\,\lambda\,|\mathcal{S}_i|G_i}
{m_i(1-\delta_i)}
\cdot\|\sigma^{(t)}-\bar\sigma^{(t)}\|_\infty\\
&\quad+ \delta_i \varepsilon_{i,Q}^{(t)}\\
&\le \frac{(1+p_i^{\max})\kappa\,|\mathcal{W}|\,
|\mathcal{Z}_{-i}|\,|\mathcal{X}_{-i}|\,
\hat{\mathcal{A}}\,\lambda\,|\mathcal{S}_i|G_i}
{m_i(1-\delta_i)}\\
&\quad\cdot\max_i\frac{\sqrt{|\mathcal{A}_i|}}{\tau_i}
\cdot\varepsilon_{Q}^{(t)}
+ \delta_i \varepsilon_{i,Q}^{(t)}
\end{align*}
Hence, by taking $\max_i$ on both sides
   \[\max_i \varepsilon_{i,Q}^{(t+1)}=\varepsilon_{Q}^{(t+1)}\le \rho\,\varepsilon_{Q}^{(t)}\]

where

\begin{align*}
\rho &:= \max_i\frac{(1+p_i^{\max})\kappa\,|\mathcal{W}|\,
|\mathcal{Z}_{-i}|\,|\mathcal{X}_{-i}|\,
\hat{\mathcal{A}}|\mathcal{S}_i|G_i}{m_i(1-\delta_i)}\\
&\quad\cdot\max_i \frac{\sqrt{|\mathcal{A}_i|}}{\tau_i}\lambda
+\max_i\delta_i.
\end{align*}
Thus, the joint Q-value iteration is a contraction 
under the max-norm with factor $\rho < 1$ 
for sufficiently small $\lambda$, completing the proof.

\end{proofOf}

\end{document}